\documentclass[12pt]{iopart}
\expandafter\let\csname equation*\endcsname\relax
\expandafter\let\csname endequation*\endcsname\relax
\usepackage{amsmath}
\usepackage{amsfonts}
\usepackage{amssymb}
  \usepackage{amsthm}
\usepackage{mathrsfs}
\usepackage{graphicx}
\usepackage{color}
\usepackage{float}
\usepackage{mathtools}
\usepackage{algorithm}
\usepackage{algpseudocode}
\usepackage{wrapfig}
\usepackage{mathbbol}
\usepackage[colorlinks, linkcolor=blue, citecolor = blue, urlcolor = blue]{hyperref}
\usepackage{cite}
\usepackage{upgreek}
\usepackage{multirow}
\usepackage{relsize}
\usepackage{orcidlink}
\usepackage[symbol]{footmisc}

\newcommand \bes {\begin{equation}\begin{split}}
\newcommand \be {\begin{equation}}
\newcommand \ee {\end{equation}}
\newcommand \ees {\end{split}\end{equation}}

\newtheorem{theorem}{Theorem}[section]

\newtheorem{prop}[theorem]{Proposition}

    \newcommand{\diff}{\mathrm{d}}

\newcommand{\bfx}{\boldsymbol{x}}

\newcommand{\bfJ}{\boldsymbol{J}}

\begin{document}

\title[A Non-Isothermal Phase-field crystal Model with Lattice Expansion]
{A Non-Isothermal Phase-Field Crystal Model with Lattice Expansion: Analysis and Benchmarks}

\author{Maik Punke$^1$\,\orcidlink{0000-0002-3564-7942},  Marco Salvalaglio$^{1,2}$\,\orcidlink{0000-0002-4217-0951}, Axel Voigt$^{1,2}$\,\orcidlink{0000-0003-2564-3697}, Steven M. Wise$^3$\,\orcidlink{0000-0003-3824-2075}}
\address{$^1$Institute of Scientific Computing, TU Dresden, 01062 Dresden, Germany.}
\address{$^2$Dresden Center for Computational Materials Science, TU Dresden, 01062 Dresden, Germany.}
\address{$^3$Department of Mathematics, The University of Tennessee, Knoxville, TN 37996, USA.}
\ead{maik.punke@tu-dresden.de}
\vspace{10pt}

\begin{abstract} 
We introduce a non-isothermal phase-field crystal model including heat flux and thermal expansion of the crystal lattice. The thermal compatibility condition, as well as a positive entropy-production property, is derived analytically and further verified  by numerical benchmark simulations. Furthermore, we examine how the different model parameters control density and temperature evolution during dendritic solidification through extensive parameter studies. Finally, we extend our framework to the modeling of open systems considering external mass and heat fluxes. This work sets the ground for a comprehensive mesoscale model of non-isothermal solidification including thermal expansion within a positive entropy-producing framework, and provides a benchmark for further meso- to macroscopic modeling of solidification.
\end{abstract}


    \section{Introduction}
Solidification of crystalline materials is a ubiquitous phenomenon with wide-ranging implications  in both nature and technology. It involves the interaction of different instabilities and out-of-equilibrium growth conditions~\cite{Langer1980}, thus being characterized by competing physical mechanisms and complex morphologies. 
Managing solidification processes is crucial for conventional technological applications~\cite{flemings1974solidification,dantzig2016solidification} as well as  self-assembly approaches~\cite{Stangl2004,polshettiwar2009self}. At early stages, crystalline seeds grow following nucleation or on pre-existing crystal phases/seeds, while later stages are affected by capillarity, elasticity, plasticity, as well as various kinetic effects~\cite{BOETTINGER200043,hoyt2003atomistic,Jaafar2017,Granasy2019,Alexandrov2021}. Throughout the entire solidification process, heat transfer within the solid and liquid phases significantly impacts the morphology of crystal growth, grain size, and phase distribution~\cite{chalmers1964principles,Jaafar2017,Granasy2019,Alexandrov2021}. On a microscopic level, the arrangement of atoms in a periodic lattice usually results in anisotropic behaviors, such as faceting, and affects the nucleation and movement of defects. These phenomena are closely related to crystallographic directions and are influenced by thermal gradients within the material.

The modeling of solidification has been addressed by various approaches resolving different time- and length-scales of interest. At the scale of individual atoms, microscopic approaches like molecular dynamics techniques~\cite{alder1959studies} describe lattice-dependent features such as anisotropies and defect structures~\cite{hoyt2003atomistic}. Due to the limited scaling properties of these models, the growth of crystals, which involves long timescales, is typically not accessible. Macroscopic approaches like sharp-interface and phase-field methods~\cite{KOBAYASHI1993410,Karma1998,Zhu2007,Steinbach2009,Pan2010,takaki2014phase,kaiser2020} coping with large systems and long timescales proved successfully in describing the main features of crystal growth. However, a direct connection to the lattice symmetry and microscopic features is not inherently and self-consistently captured in general and needs to be included through parameters and additional functions, e.g., as anisotropic interface energies~\cite{SUZUKI2002125,Tor2009,salvalaglio2015faceting}.

The so-called phase-field crystal (PFC) model~\cite{Elder2002,Elder2004,Provatas2010,Emmerich2012}
emerged as a prominent approach to describe crystal structures at large (diffusive) timescales through a continuous, periodic order parameter representing the atomic density. It reproduces the main phenomenology for crystalline systems, including solidification and crystal growth, capillarity, lattice deformations as well as nucleation and defect kinematics. The PFC model describes self-consistently anisotropies resulting from the lattice structure~\cite{Backofen_2009,PODMANICZKY2014148,Ofori-Opoku2018} and inherently includes  elasticity effects consistent with continuum mechanics with high accuracy reached when considering dedicated PFC model extensions~\cite{skogvoll2022hydrodynamic,HeinonenPRE2014,HeinonenPRL2016,SkaugenPRL2018,SkaugenPRB2018,SalvalaglioNPJ2019,chockalingam2021elastic}.
Therefore, the PFC concept represents a comprehensive framework for the description of crystalline materials in two- and three-dimensions~\cite{Tegze2011,TANG2011146,Emmerich2012}.

Only recently, further extensions of the PFC model have been developed to explicitly describe temperature and heat flux phenomena~\cite{Kocher2019,wang2021thermodynamically,punke2022explicit,punke2023improved,burns2024two}. However, thermal expansion or compression of the equilibrium crystal lattice length due to temperature fluctuations is either not considered~\cite{Kocher2019,wang2021thermodynamically,burns2024two} or a positive entropy-production rate cannot be guaranteed~\cite{punke2022explicit}.

This work presents a non-isothermal PFC model with thermal expansion of the crystal lattice within a positive entropy-producing framework. As analytical benchmarks, we derive the thermal compatibility condition as well as a positive entropy-production property in section~\ref{sec:analysis} which is further verified by  
numerical benchmark simulations in section~\ref{sec:benchmarks}.  Furthermore, the role of different model parameters is examined through parameter studies focusing on dendritic solidification, as shown in section~\ref{sec:parameterStudy}. Finally, we extend our framework to the description of external mass and heat fluxes entering the system via the domain boundary, mimicking the physics of open systems as illustrated in section~\ref{sec:openSystems}. The derived analytical results, along with the proposed simulation setups, offer valuable benchmarks for PFC and other solidification modeling approaches that incorporate heat flux in both open and closed systems.

\section{Model introduction and analysis}
\label{sec:analysis}
The phase-field crystal (PFC) model~\cite{Elder2002,Elder2004,Emmerich2012}
describes crystal structures at diffusive timescales through a continuous, periodic order parameter $\psi \equiv \psi(\bfx,t)$ representing in its dimensionless formulation the atomic number difference with respect to the liquid phase. It is based on a Swift-Hohenberg-like free energy functional which can be written 
    \begin{equation}
F\left[\psi\right] = \int_\Omega \left( \dfrac{\lambda-\kappa}{2}\psi^2 -  \delta \dfrac{\psi^3}{6} + \dfrac{\psi^4}{12}+ \dfrac{\kappa}{2}\psi\, \mathcal{L}\, \psi\right) \,\diff\bfx,   
    \label{eq:F1}
    \end{equation}
with $\delta,\,\lambda,\,\kappa \geq 0$ parameters characterizing the phase space and material properties together with the global average density $\Psi=\frac{1}{|\Omega|}\int_\Omega \psi\,  \rm{d}\bfx$, and $\Omega \in \mathbb{R}^n$ the domain of definition of $\psi$ with $n=2$ defined here as $\Omega =[-L_x/2,L_x/2]\times [-L_y/2,L_y/2]$. The differential operator $\mathcal{L}$ approximates a two-point correlation function and thus encodes the crystal symmetry. For instance, 2D triangular symmetry is obtained as equilibrium state for some free-energy parameters with $\mathcal{L}=(1+\nabla^2)^2$~\cite{Emmerich2012}. 

Together with appropriate boundary and initial conditions, the dynamical equation for $\psi$ can be generally described by a conservative gradient flow of $F$, that is,
    \begin{equation}
    \begin{split}
\partial _t \psi &= \nabla\cdot\left(M(\psi) \nabla\dfrac{\delta F\left[\psi\right]}{\delta \psi}\right),
    \end{split}
    \label{eq:pfc}
    \end{equation}
where $M(\psi)>0$ is a mobility function. Further extensions of~\eqref{eq:F1}-\eqref{eq:pfc} may be readily considered to account for other lattice symmetries~\cite{Greenwood2010,Greenwood2011,Mkhonta2013} (in both 2D and 3D). In the following, we extend the classical PFC model to a non-isothermal framework, including thermal lattice expansion, that satisfies a positive entropy-production property.

    \subsection{Energy and Entropy Densities}
Let $T$ be the dimensionless temperature field with $T=1$ representing the melting point. It is assumed that the free energy $F$, entropy $S$, and internal energy $E$ have the generic forms
\begin{equation*}
    \begin{split}
        &F = \int_\Omega\hat{f}(\psi,\nabla \psi, \nabla^2 \psi, T)\,\diff\bfx,\\
        &S = \int_\Omega\hat{s}(\psi,\nabla \psi, \nabla^2 \psi, T)\,\diff\bfx,\\
        &E = \int_\Omega\hat{e}(\psi,\nabla \psi, \nabla^2 \psi, T)\,\diff\bfx.
    \end{split}
\end{equation*}
In Ref.~\cite{punke2022explicit}, we introduced a PFC model with a free energy featuring explicit coupling with the temperature field under periodic boundary conditions for both $\psi$ and $T$. In the current work, we start with the same free energy as in~\cite{punke2022explicit}, since its functional form is well understood, and then derive the entropy and internal energy densities $\hat{s}$ and $\hat{e}$ subsequently. The free energy density $\hat{f}$ is written as
    \begin{equation}
\hat{f}(\psi,\nabla \psi, \nabla^2 \psi, T) = f(\psi,T) + \frac{\kappa T}{2}\left(\alpha^2(T) \psi^2- 2  \alpha(T) \left| \nabla \psi\right|^2 +\left(\nabla^2 \psi\right)^2 \right) ,
    \label{eqn:free-energy-density-pfc}
    \end{equation}
with
    \[
f(\psi,T) = k_{\rm B} T \left(\dfrac{\lambda-\kappa}{2}\psi^2 -  \delta \dfrac{\psi^3}{6} + \dfrac{\psi^4}{12}\right)
- C_vT\log\left(\frac{T}{T_0}\right) -\beta \psi,
    \]
where $k_{\rm B}$ is Boltzmann's constant, hereafter set to $k_{\rm B}\equiv 1$ for simplicity, $C_v>0$ is the constant heat capacity of the material, $\beta>0$ is a constant, and $\alpha(T)>0$ measures the degree of thermal expansion and has physical units (before non-dimensionalization) of $\mathrm{length}^{-2}$. Here, we use the explicit model
   \begin{equation}
\alpha(T) =  \frac{1}{\left(a_0 + a_1 (T - T_0) \right)^2}.
    \label{eq:thermalExpansion}
    \end{equation}
$a_0>0$ is a reference length (set to $a_0=1$), and $a_1\ge 0$ is the thermal expansion coefficient, describing an expansion or compression of the crystal lattice, depending upon whether the temperature $T$ is larger or smaller than the reference temperature, $T_0 >0$. When thermal expansion is neglected ($\alpha\equiv\alpha_0>0$),  the form of the free energy assumed in~\cite{wang2021thermodynamically} is recovered.
We use the standard thermodynamic assumptions that 
    \begin{equation}
\hat{s} = -\partial_T\hat{f} \quad \mbox{and} \quad \hat{e} = \hat{f}+T\hat{s} .
    \label{eqn:thermo-4}
    \end{equation}
The second condition of \eqref{eqn:thermo-4} encodes a basic assumption of non-equilibrium thermodynamics, namely, that the energy, entropy, and free energy are related locally as if the system were near equilibrium at every point.

As a consequence of \eqref{eqn:thermo-4}, we can derive the following thermal compatibility condition, which will be exploited later:
    \begin{prop}
    \label{prop:thermal-comp-2}
Let the free energy density $\hat{f}$, defined as in \eqref{eqn:free-energy-density-pfc}, be a twice continuously differentiable function of temperature, $T$. Suppose that the entropy and internal energy densities $\hat{s}$ and $\hat{e}$ are computed as in \eqref{eqn:thermo-4}. Then,
    \begin{equation}
\frac{1}{T} \frac{\partial \hat{e}}{\partial T} = \frac{\partial\hat{s}}{\partial T} .
    \label{eqn:thermal-comp-2}
    \end{equation}
    \end{prop}

    \begin{proof}
Using \eqref{eqn:thermo-4}, we have
    \begin{align*}
\partial_T\hat{e} & = \partial_T\hat{f} + \hat{s} + T\partial_T\hat{s}
    \\
& = -\hat{s} + \hat{s} + T\partial_T\hat{s}
    \\
& = T\partial_T\hat{s},
    \end{align*}
which easily  proves the result.
    \end{proof}
Using the explicit form of $\hat{f}$ from equations~\eqref{eqn:free-energy-density-pfc}, the entropy density $\hat{s}$ and the energy density $\hat{e}$ can be written as
    \begin{align*}
\hat{s} =& -\partial_T f - \frac{\kappa  }{2}\left(\alpha^2(T) \psi^2- 2 \alpha(T) \left| \nabla \psi\right|^2 +\left(\nabla^2 \psi\right)^2 \right) \dots 
    \\
& - \frac{\kappa T}{2}\left(2\alpha(T)\alpha'(T) \psi^2- 2\alpha'(T) \left| \nabla \psi\right|^2 \right)  ,\\
   \hat{e} =& f-T \partial_T f - \frac{\kappa T^2}{2} \left( 2\alpha(T) \alpha'(T) \psi^2 - 2\alpha'(T) \left| \nabla \psi\right|^2 \right),
 \end{align*}
leading to     
    \[
\hat{e} = C_vT - \beta \psi  - \frac{\kappa T^2}{2} \left( 2\alpha(T) \alpha'(T) \psi^2 - 2\alpha'(T) \left| \nabla \psi\right|^2 \right).
    \]
In those cases where the thermal expansion may be neglected, that is, $\alpha'(T) \equiv 0$, as in~\cite{wang2021thermodynamically}, the internal energy density $\hat{e}$ simplifies to
    \[
\hat{e} = C_vT - \beta \psi.
    \]
For the general case, we make the following definitions:   
    \begin{equation}\label{eq:gammas}
\gamma_0 (T) := \kappa  T^2 \alpha(T) \alpha'(T) \quad \mbox{and} \quad  \gamma_1(T) := \kappa T^2\alpha'(T).
    \end{equation}
The internal energy density can then be expressed as 
    \[
\hat{e} = C_vT - \beta \psi  - \gamma_0(T) \psi^2 + \gamma_1(T) \left| \nabla \psi\right|^2 .
    \]

    \subsection{Mass and Internal Energy Conservation}

We assume that the flow of heat is dominated by diffusive flux. The first law of thermodynamics, which states that energy must be conserved, then takes the form 
    \begin{equation}
\partial_t\hat{e} = -\nabla\cdot\bfJ_e,
    \label{eqn:energy-conservation}
    \end{equation}
where  $\bfJ_e$ is the flux of energy, here chosen as
    \[
\bfJ_e = \widehat{M}_T \nabla\left( \frac{1}{T} \right),
    \]
where $\widehat{M}_T>0$ is a thermal energy mobility coefficient so that the energy conservation equation is
    \begin{align}
C_v\partial_tT -\beta\partial_t \psi  =& -\nabla\left( \widehat{M}_T \nabla\left(\frac{1}{T}\right)\right) + \gamma_0'(T) \psi^2 \,  \partial_t T - \gamma_1'(T) |\nabla \psi|^2 \, \partial_tT
    \nonumber
    \\
& + 2 \gamma_0(T) \psi \,  \partial_t \psi - 2\gamma_1(T) \nabla \psi \cdot \nabla \partial_t \psi.
    \label{eqn:energy-con-2}
    \end{align}
Based on the general form of the energy equation, it follows directly that the first law of thermodynamics holds globally, $\partial_t E = 0$. 

The generic form of the mass conservation equation is
    \[
\partial_t  \psi  = -\nabla\cdot\bfJ_\psi,
    \]
with mass flux $\bfJ_\psi$, imposing $\partial_t \int_\Omega \psi(\bfx,t)\, \diff\bfx= 0$. For  $\bfJ_\psi$, we assume  $\bfJ_\psi=-M_\psi  \nabla w$  with  mass mobility coefficient $M_\psi>0$ and generalized chemical potential $w$:
    \begin{align}
\partial_t  \psi & = \nabla\cdot \left(M_\psi  \nabla w\right), & \mbox{(mass conservation)} &
    \label{eqn:mass-con-3}
    \\
w & = \partial_\psi \left(\frac{\hat{f}}{T}\right) - \nabla\cdot\left(\partial_{\nabla \psi}\left(\frac{\hat{f}}{T}\right)  \right) +\kappa  \nabla^4 \psi, & \mbox{(generalized chemical potential)} &
    \label{eqn:mass-con-4}
    \end{align}
We observe that
    \[
\frac{\hat{f}}{T} = \widetilde{f}(\psi) - C_v\log\left(\frac{T}{T_0}\right) -\beta \frac{\psi}{T}  + \frac{\kappa }{2}\left(\alpha^2(T) \psi^2- 2  \alpha(T) \left| \nabla \psi\right|^2 +\left(\nabla^2 \psi\right)^2 \right),
    \]
where
    \[
\widetilde{f}(\psi) := \dfrac{\lambda-\kappa}{2}\psi^2 -  \delta \dfrac{\psi^3}{6} + \dfrac{\psi^4}{12},
    \]
and, therefore, the generalized chemical potential $w$ can be specified as
    \begin{align*}
w &= \partial_\psi \left(\frac{\hat{f}}{T} \right) \, \partial_t \psi  - \nabla\cdot \left(\partial_{\nabla \psi} \left(\frac{\hat{f}}{T}\right) \right) + \nabla^2\left( \partial_{\nabla^2 \psi}\left( \frac{\hat{f}}{T}\right)\right)
    \\
& = \widetilde{f}'(\psi) - \frac{\beta}{T} +\kappa  \alpha^2(T) \psi  + 2 \kappa  \nabla\cdot \left( \alpha(T) \nabla \psi \right) + \kappa  \nabla^4 \psi .
    \end{align*}
In the following subsection, we show that the flux choices for $\bfJ_\e$ and $\bfJ_\psi$ lead to a thermodynamically consistent model.

    \subsection{Entropy Production}
    
    \begin{prop}
Assume that $\Omega$ is a rectangular domain in $\mathbb{R}^2$ and the fields are $\Omega$-periodic. With the constitutive choices above, the entropy is increasing in time, with the following positive entropy production rate:
    \begin{equation}
\partial_t S = \int_\Omega\left\{ M_\psi  \left| \nabla w \right|^2   + \hat{M}_T \left| \nabla\left( \frac{1}{T} \right) \right|^2 \right\} \diff \bfx \ge 0.
     \end{equation}
     \label{eq:entropyProduction}
    \end{prop}
    \begin{proof}
Taking the time derivative of the total entropy and using thermal compatibility~\eqref{eqn:thermal-comp-2} and the energy conservation equation~\eqref{eqn:energy-conservation}, we have
    \begin{align*}
\partial_t S  =& \int_\Omega \left\{ \partial_\psi \hat{s} \, \partial_t \psi + \partial_T \hat{s} \, \partial_t T + \partial_{\nabla \psi} \hat{s} \cdot \partial_t \nabla \psi+ \partial_{\nabla^2 \psi} \hat{s} \, \partial_t(\nabla^2 \psi) \right\} \diff\bfx 
    \\
  =& \int_\Omega \left\{ \partial_\psi \hat{s} \, \partial_t \psi + \frac{1}{T} \partial_T \hat{e} \, \partial_t T + \partial_{\nabla \psi} \hat{s} \cdot \nabla\partial_t \psi + \partial_{\nabla^2 \psi} \hat{s} \, \nabla^2 \partial_t \psi \right\} \diff\bfx 
    \\
  =& \int_\Omega \left\{ \partial_\psi \hat{s} \, \partial_t \psi  + \partial_{\nabla \psi} \hat{s} \cdot \nabla\partial_t \psi + \partial_{\nabla^2 \psi} \hat{s} \, \nabla^2 \partial_t \psi \right\} \diff\bfx \dots 
    \\
 & + \int_\Omega \left\{ \frac{1}{T} \left( -\nabla\cdot\bfJ_e - \partial_\psi\hat{e} \, \partial_t \psi - \partial_{\nabla \psi}\hat{e} \cdot\nabla\partial_t \psi - \partial_{\nabla^2 \psi} \hat{e} \, \nabla^2\partial_t \psi  \right) \right\} \diff\bfx 
    \\
  =& -\int_\Omega \left\{ \partial_\psi \left(\frac{\hat{f}}{T} \right) \, \partial_t \psi  + \partial_{\nabla \psi} \left(\frac{\hat{f}}{T}\right) \cdot \nabla\partial_t \psi + \partial_{\nabla^2 \psi}\left( \frac{\hat{f}}{T}\right) \, \nabla^2 \partial_t \psi  + \frac{1}{T} \nabla\cdot\bfJ_e  \right\} \diff\bfx.
    \end{align*}
Using periodic boundary conditions, integration-by-parts, and mass conservation, we have
    \begin{align*}
\partial_t S  & = \int_\Omega \left\{ -\left( \partial_\psi \left(\frac{\hat{f}}{T} \right) \, \partial_t \psi  - \nabla\cdot \left(\partial_{\nabla \psi} \left(\frac{\hat{f}}{T}\right) \right) + \nabla^2\left( \partial_{\nabla^2 \psi}\left( \frac{\hat{f}}{T}\right)\right) \right) \partial_t \psi  - \frac{1}{T} \nabla\cdot\bfJ_e  \right\} \diff\bfx
    \\
& = \int_\Omega \left\{ w \nabla\cdot \bfJ_\psi  - \frac{1}{T} \nabla\cdot\bfJ_e  \right\} \diff\bfx
    \\
& = \int_\Omega \left\{ -\nabla w \cdot \bfJ_\psi  + \nabla\left( \frac{1}{T}\right)\cdot\bfJ_e  \right\} \diff\bfx.
    \end{align*}
Finally, using the flux conditions that appear in equations~\eqref{eqn:energy-con-2}--\eqref{eqn:mass-con-4}, namely
    \[
\bfJ_\psi = - M_\psi \nabla w \quad \mbox{and} \quad \bfJ_e =\widehat{M}_T \nabla\left( \frac{1}{T} \right),
    \]
we achieve the desired rate of entropy production:
   \[
\partial_t S = \int_\Omega\left\{ M_\psi  \left| \nabla w \right|^2   + \hat{M}_T \left| \nabla\left( \frac{1}{T} \right) \right|^2 \right\} \diff \bfx \ge 0.
    \]
    \end{proof}

    \subsection{PFC Model Recapitulation}

As stated earlier, we use periodic boundary conditions for simplicity. We assume a constant mass mobility coefficient $M_\psi>0$,  and we suppose that the thermal energy mobility coefficient satisfies $\widehat{M}_T = T^2{M}_T$, where ${M}_T>0$ is a constant.
The complete model results in a coupled system of equations that includes a heat-like equation and a PFC-like equation:
   \begin{equation}
       \begin{split}
        M_T\nabla^2 T  &=\left(C_v - \gamma_0'(T) \psi^2 + \gamma_1'(T) |\nabla \psi|^2\right)\partial_tT  -\left( \beta + 2 \gamma_0(T) \psi - 2\gamma_1(T) \nabla \psi \cdot \nabla\right) \partial_t \psi, \\
           \partial_t  \psi  &= M_\psi \nabla^2 \left(
  (\lambda-\kappa)\psi-\delta\dfrac{\psi^2}{2}+\dfrac{\psi^3}{3} - \frac{\beta}{T} +\kappa  \alpha^2(T) \psi   + 2 \kappa  \nabla\cdot \left( \alpha(T) \nabla \psi \right) + \kappa  \nabla^4 \psi  \right)
       \end{split}
       \label{eq:pfcT}
   \end{equation}
where $\gamma_0 (T)$ and $\gamma_1 (T)$ defined as in equation~\eqref{eq:gammas} and $\alpha(T)$ as in equation~\eqref{eq:thermalExpansion}. 
In addition to those parameters that are essential for the standard isothermal PFC model, namely, the parameters $\delta, \lambda, \kappa \geq 0$ and the global average density $\Psi$, our proposed non-isothermal model~\eqref{eq:pfcT} includes mobility constants $M_\psi, M_T>0$, the material constant $\beta>0$ controlling the coupling of the PFC equation with the heat flux, the heat capacity $C_v>0$ and a thermal expansion coefficient $a_1\geq 0$ which enters $\alpha (T)$, encoding how the crystal lattice expands due to temperature fluctuations. We will investigate in the following section~\ref{sec:parameterStudy} how the individual model parameters affect crystal growth within closed systems and extend our formulations to the modeling of open systems in section~\ref{sec:openSystems}. Table~\ref{tab:numerics} summarizes all model parameters for the simulations reported in this paper.

\begin{table}[h!]
    \centering
    \resizebox{\textwidth}{!}{
    \begin{tabular}{llllllllll}
    \hline
         Figure&$\Psi$&  $C_v$& $M_T$ &$M_\psi$& $\beta$& $a_1$ & $T_0$\\
         \hline
        \ref{fig:dendrite} &0.849, 0.151&0.06&0.06&1&0.06 &0.1&0.6\\
        \ref{fig:dynamics} &&$0.05\dots 2$&&&&&\\
            &&0.06&$2\cdot 10^{-3}\dots 0.3$&&&&\\
             &&&0.06&$0.01\dots 100$&&&\\
       \ref{fig:flip} &&&&1&$6\cdot 10^{-4}\dots 0.3$&&\\  
        &&&&&0.06&$0.01\dots 0.3$&\\
            &&&&&&0.1&$0.1\dots 1.8$\\
           \ref{fig:openS}  &-&&&&&0.01&0.6\\
                       \hline  \end{tabular}}
    \caption{Model parameters
    for all the simulations reported in this paper. We set $\lambda = 0.6$, $\kappa = 0.46$, and $\delta = 1$. For the setup according to figure~\ref{fig:dendrite},~\ref{fig:dynamics} and~\ref{fig:flip}  the domain is set to $\Omega = 220\times 256$UC, while for the setup according to figure~\ref{fig:openS}, we choose  $\Omega = 441\times 128$UC with 1 unit cell (UC) representing the smallest repeat unit of $\psi$:  1UC $= [0,p_x]\times [0,p_y]$ and $p_x=2/\sqrt{3} p_y= 4\pi/\sqrt{3}$. Uniform spatial and temporal grids are chosen to guarantee numerical convergence and positive entropy production (for all simulations we choose $\Delta x \approx \Delta y \approx 1$ with $\Delta t = 0.01$ for the setup according to figure~\ref{fig:dendrite},~\ref{fig:dynamics},~\ref{fig:flip} and $\Delta t = 0.1$ for the setup according to figure~\ref{fig:openS}).  Empty table entries read as the row above.}
    \label{tab:numerics}
\end{table}

\section{Numerical benchmark simulations}

Numerical solutions of the proposed model can be obtained efficiently by using a Fourier pseudo-spectral method, enforcing periodic boundary conditions in combination with a linear first-order semi-implicit (IMEX) time-stepping scheme. The numerical algorithm is implemented on graphics process units (GPUs), and we perform the (inverse-) Fourier transforms using a cuFFT library.

\label{sec:benchmarks}
\subsection{Model features and parameter studies}
\label{sec:parameterStudy}
\begin{figure}[h]
\centering
\includegraphics[width=\textwidth]{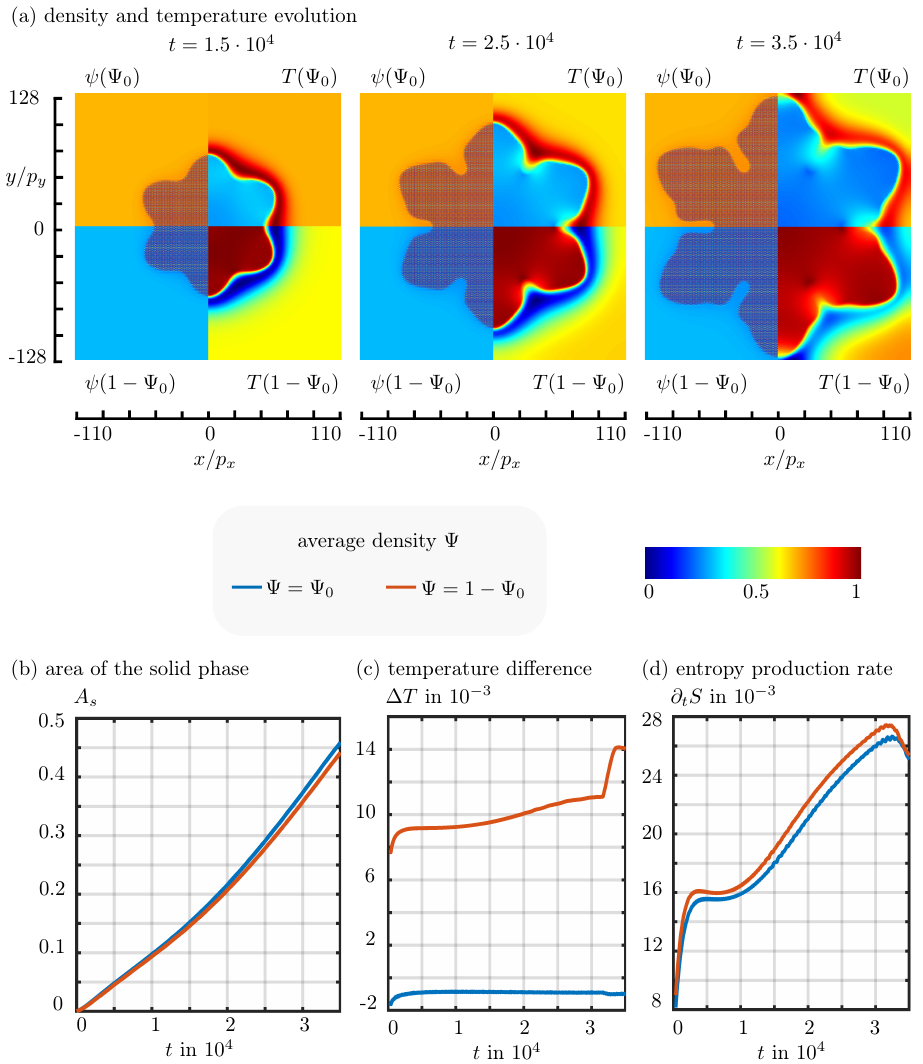}
\caption{Dendritic solidification by the proposed PFC model. The symmetry break for average densities $\Psi=\Psi_0$ and $\Psi=1-\Psi_0$ during dendritic solidification is also illustrated. (a) density ($\psi$) and temperature ($T$) evolution for $\Psi=\Psi_0$ (top) and $\Psi=1-\Psi_0$ (bottom) during growth at different time steps. (b) relative area of the solid phase $A_{\rm s}$ vs. time (c) signed temperature difference $\Delta T$ vs. time. (d) entropy production rate $\partial_t S$ vs. time.   For illustration purposes, we normalized all plotted quantities from 0 to 1, values of $\psi(\Psi_0)$, $\psi(1-\Psi_0)$, $T(\Psi_0)$ and $T(1-\Psi_0)$ vary within the following ranges: $-0.03\leq \psi(\Psi_0)\leq 1.21$, $-0.02\leq \psi(1-\Psi_0)\leq 0.97$, $0.6-6\cdot 10^{-4} \leq T(\Psi_0) \leq 0.6+6\cdot 10^{-4}$ and $0.6-1.0\cdot 10^{-2}\leq T(1-\Psi_0)\leq 0.6+3.6\cdot 10^{-3}$.  Lengths are scaled with the atomic spacings along the $x$- and $y$-axis, $p_x$ and $p_y$.}
	\label{fig:dendrite}
    \end{figure}

We focus on a benchmark setup of a growing crystal in a supersaturated melt with initially constant temperature $T_0$ to investigate the density and temperature evolution in a closed system for different model parameter values, similar to~\cite{punke2022explicit}. 
We choose our model- and material parameters as reported in table~\ref{tab:numerics} and compare the resulting density and temperature evolution for the global average densities $\Psi=\Psi_0(=0.849)$ and  $\Psi=1-\Psi_0(=0.151)$ in figure~\ref{fig:dendrite}, representing crystals with underlying honeycomb and triangular symmetry, respectively. For the chosen values of $\Psi$, the melt is mass-supersaturated, enabling solidification of the growing crystal within a closed system. Representative steps of the dendritic solidification process are shown in figure~\ref{fig:dendrite}(a) with snapshots of the density and temperature fields. 
Relative to the commonly observed enlarged lattice spacing produced by PFC models during growth~\cite{punke2022explicit,punke2023evaluation}, further compression of the crystal lattice is obtained for 
 heat reduction inside the solid phase ($\Psi=\Psi_0$) and a further  expansion of the crystal lattice is obtained for heat production inside the solid phase ($\Psi=1-\Psi_0$), as described by equation~\eqref{eq:thermalExpansion}.
Next to a qualitatively different temperature evolution inside the solid phase for $\Psi=\Psi_0$ and  $\Psi=1-\Psi_0$, an additional (small) symmetry break in the dendritic morphologies is observed (slightly faster solidification dynamics corresponding to smaller temperature gradients for  $\Psi= \Psi_0$, further analysis below).
As first described in~\cite{punke2022explicit}, only in case $a_1=0$, identical energetics for $\Psi=\Psi_0$ and $\Psi=1-\Psi_0$  lead to identical density and temperature morphologies. To further inspect this symmetry break, we quantify the dendritic growth process by calculating the relative area of the solid phase $A_{\rm s}$ (area of the solid phase divided by the domain size) and show its time evolution in figure~\ref{fig:dendrite}(b), illustrating a slightly faster growth of the solid phase for $\Psi=\Psi_0$. Additionally, we calculate the signed temperature difference within the system as $\Delta T(t) = \left| \max_{\,\bfx \in \Omega} T(\bfx,t) -\min_{\,\bfx \in \Omega} T(\bfx,t) \right| \mathrm{sign} \left(  T(\bfx=0,t)-T_0 \right)$, leading to positive values for a solidification process with heat development inside the solid phase (e.g., for $\Psi = 1-\Psi_0$) and to negative values for a solidification process with heat reduction inside the solid phase (e.g., for $\Psi=\Psi_0$). In figure~\ref{fig:dendrite}(c), we show the time evolution of $\Delta T$ and observe a larger temperature difference within the system for $\Psi=1-\Psi_0$, corresponding to slower solidification dynamics.
The different energetics for $\Psi=\Psi_0$ and $\Psi=1-\Psi$ lead to different positive entropy-production rates $\partial_t S$~\eqref{eq:entropyProduction}, which we show in
figure~\ref{fig:dendrite}(d). We note that after the initial transient behavior for $t\lesssim 2\cdot 10^3$, the entropy production rates for $\Psi=\Psi_0$ and $\Psi=1-\Psi_0$ remain unchanged until they increase for $t\gtrsim 1\cdot 10^4$ due to the developing dendritic arms. The interaction of the density and temperature fields with their periodic images leads to a decrease in $\partial_t S$ for $t\gtrsim 3.1\cdot 10^4$, see figure~\ref{fig:dendrite}(d).

We further investigate the role of different parameters entering the proposed PFC model, equation~\eqref{eq:pfcT}, during the dendritic solidification process for $\Psi=\Psi_0$ and $\Psi=1-\Psi_0$ through numerical parameter studies. We vary one of the following parameters while letting the others unchanged: $M_\psi$, $M_T$, $C_v$, $a_1$, $\beta$ and $T_0$, see table~\ref{tab:numerics}. We keep the setup from figure~\ref{fig:dendrite}, calculate $A_{\rm s}$, $\Delta T$ at $t=3.5\cdot 10^4$ for both $\Psi = \Psi_0 $ and  $\Psi = 1-\Psi_0$ and show their dependence on the different model parameters in figure~\ref{fig:dynamics} and~\ref{fig:flip}. 

For larger parameter values of $M_\psi$, $M_T$ and lower parameter values of $C_v$, an increase in $A_{\rm s}$ is obtained, corresponding to a faster solidification of the growing dendrites, see figure~\ref{fig:dynamics}(a). The temperature field inside the growing solid phase can be controlled quantitatively by varying  $M_\psi$, $M_T$, and $C_v$, but does not change its qualitative behavior (heat production inside the solid phase for $\Psi=1-\Psi_0$ and heat reduction inside the solid phase for $\Psi=\Psi_0$), see figure~\ref{fig:dynamics}(b). We conclude that $M_\psi$, $M_T$, and $C_v$ mainly control the time scales of our model. 

By varying the parameter values of $a_1$, $\beta$, and $T_0$, the relative area of the growing solid phase $A_{\rm s}$ shows a maximum for $\Psi=\Psi_0$ and $\Psi=1-\Psi_0$, respectively,  see figure~\ref{fig:flip}(a). The maximal values of $A_{\rm s}$ correspond to parameter values $a_1$, $\beta$ and $T_0$ which lead to vanishing temperature differences $(\Delta T\approx 0)$, see figure~\ref{fig:flip}(b). 
Furthermore, a changed qualitative behavior of the temperature field inside the growing solid phase is obtained for large parameter values of $a_1$, $T_0$, and small parameter values of $\beta$  (heat reduction inside the solid phase for $\Psi=1-\Psi_0$ and heat production inside the solid phase for $\Psi=\Psi_0$), see figure~\ref{fig:flip}(b). This qualitative change in the temperature fields corresponds to a flipped sign of the prefactor for $\partial_t \psi$ within the heat-like equation~\eqref{eq:pfcT} for large parameter values of $a_1$, $T_0$ and small parameter values of $\beta$.

We numerically verified that a positive entropy-production rate is obtained for all simulations in figure~\ref{fig:dendrite},~\ref{fig:dynamics} and~\ref{fig:flip}.

\begin{figure}[h!]
\centering
\includegraphics[width=\textwidth]{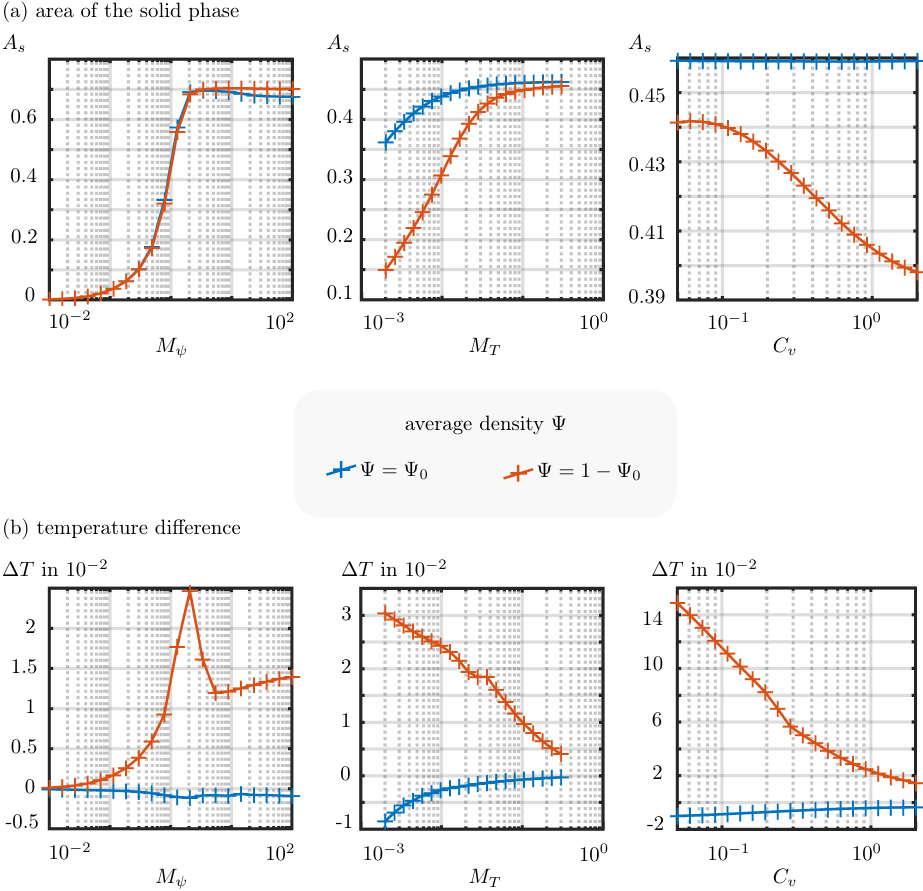}
\caption{Dependence on model parameters -- I. The mobility parameters $M_\psi$, $M_T$ and $C_v$ mainly control the dynamics of the solidification process for the setup in figure~\ref{fig:dendrite}: (a) area of the solid phase $A_{\rm s}$ and (b) temperature difference $\Delta T$, by varying $M_\psi$, $M_T$ and $C_v$.}
	\label{fig:dynamics}
    \end{figure}

\begin{figure}[h!]
\centering
\includegraphics[width=\textwidth]{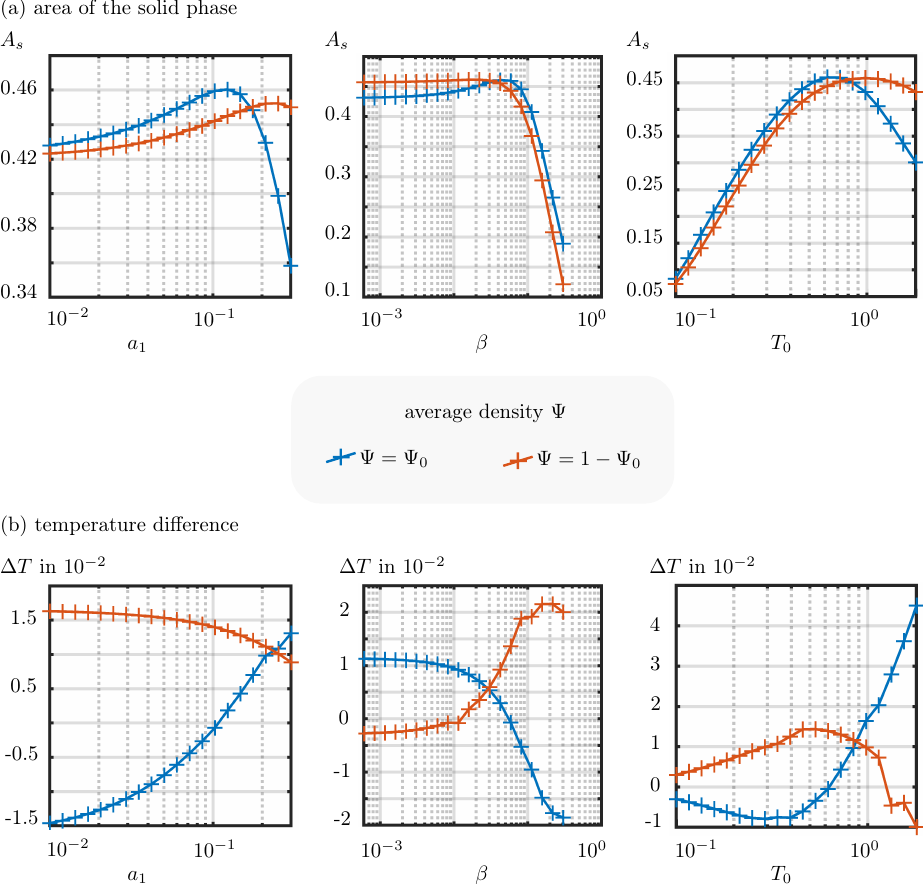}
\caption{Dependence on model parameters -- II. 
The material parameters $a_1$, $\beta$ and $T_0$ control the dynamics of the solidification process and  the qualitative behavior of the associated temperature field for the setup in figure~\ref{fig:dendrite}: (a) area of the solid phase $A_{\rm s}$ and (b) temperature difference $\Delta T$, by varying $a_1$, $\beta$ and $T_0$.}
	\label{fig:flip}
    \end{figure}
    
    \subsection{Open systems}
    \label{sec:openSystems}
To properly describe solidification in open systems, we need to modify our model to mimic external mass and heat fluxes entering the system via the domain boundary.
We will show that the modified model (described below) recovers qualitative behavior typical of the directionality of mass fluxes as well as (de-)stabilization effects of the growth front induced by heat fluxes during solidification, thus describing more realistic solidification conditions than the supersaturation condition typically used for closed systems, see section~\ref{sec:parameterStudy}.

We initialize a perturbed crystal front surrounded by liquid with  $\psi\equiv 0.87$  for $|x|<33p_x$, see figure~\ref{fig:openS}(a), and material parameters from table~\ref{tab:numerics}. The specific choice of the material parameters guarantees a local equilibrium of the solid and liquid crystal phase for $|x|<33p_x$. For  $|x|>193p_x$ we initialize a liquid crystal phase with lower density $\psi\equiv 0.86$. For $33p_x\leq |x|\leq 193p_x$ the density of the liquid crystal phase decreases linearly from $\left. \psi \right|_{|x|=33p_x}=0.87$ to  $\left. \psi\right|_{|x|=193p_x}=0.86$. 
At every time step, we adjust the density for $|x|\geq 193p_x$ setting $\left. \psi\right|_{|x|\geq 193p_x}=0.86$, thus effectively introducing a net mass flux into the system.
This adjustment directs mass towards the central crystal front, analogous to liquid diffusion observed in experimental setups~\cite{chalmers1964principles,dantzig2016solidification,cole1971inhomogeneities}. Since the solid and liquid crystal phases in the central region are initialized in equilibrium, the additional mass flux leads to an incipient solidification at the central crystal front. 

Furthermore, the stability of the growing crystal front is expected to depend on the presence (and sign) of temperature gradients, see e.g. \cite{losert1998evolution}. To test the capabilities of the approach, we therefore consider nonuniform temperatures and explore regimes featuring opposite temperature gradients. Like the density initialization described above, we set $\left.T \right| _{|x|<33p_x}\equiv T_0$ and adjust  the temperature for  $|x|>193p_x$ setting $\left. T\right|_{|x|\geq 193p_x}\equiv T_b$ at every time step. In case the applied temperature gradient is set opposite to the growth direction of the solid phase ($T_0>T_b=0.25$), a destabilization of the evolving solid-liquid interface leading to dendritic solidification is observed, see~figure~\ref{fig:openS}(b1),~(c1). When enforcing a temperature gradient along the growth direction ($T_0<T_b=0.95$), a stabilization of the initially perturbed solid-liquid interface is obtained, leading to a smoother growth front and slower solidification dynamics, see figure~\ref{fig:openS}(b2),~(c2). This result is consistent with theoretical~\cite{du2014phase,tang2009phase,kim2001phase,fu2020thermoelectric} and, importantly, experimental observations~\cite{losert1998evolution,chang2012rapid,cao2011effect,song2011dendritic,yang2011dependence}.

\begin{figure}[h!]
\centering
\includegraphics[width=\textwidth]{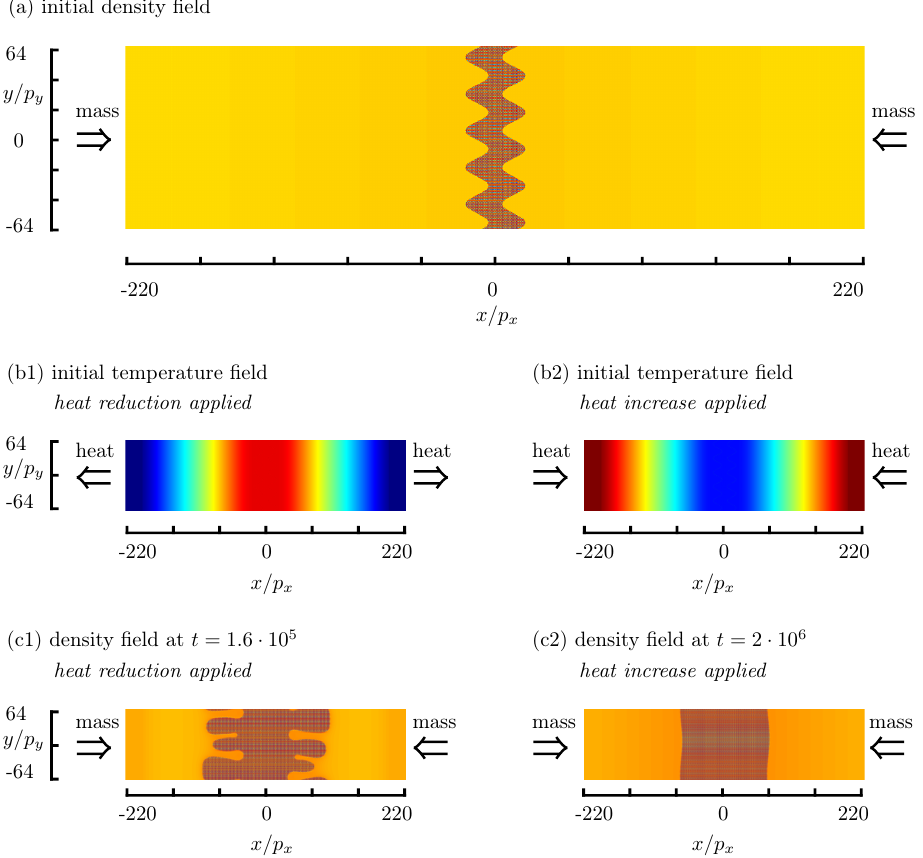}
\caption{Example of growth in an open system and interplay with the heat flux. (a) initial condition corresponding to a perturbed crystal morphology. Mass flux entering the system, inducing crystallization, is considered. (b1) initial temperature field enforcing heat reduction, leading to a dendritic-like solid-liquid interface evolution shown in (c1) by the microscopic density at $t=1.6\cdot 10^5$. (b2) initial temperature field enforcing heat increase, leading to a smooth growing front shown in (c2) by the microscopic density at $t=2\cdot 10^6$.
For illustration purposes, we normalized all plotted quantities from 0 to 1 using the color-scale from figure~\ref{fig:dendrite}, values of $\psi$ and  $T$ vary within the following ranges: (a) $0.41\leq \psi\leq 1.09$, (b1) $0.25\leq T \leq 0.6$, (c1) $0.02\leq \psi\leq 1.21$, (b2) $0.6\leq T \leq 0.95$, (c2) $0.03\leq \psi\leq 1.21$.  Lengths are scaled with the atomic spacings along the $x$- and $y$-axis, $p_x$ and $p_y$.}
	\label{fig:openS}
    \end{figure}

\section{Conclusion}
\label{sec:conclusions}
We introduced a non-isothermal PFC model that accounts for lattice expansion and satisfies a positive entropy-production property. Thermodynamic consistency is shown to follow by construction from the chosen formulation of the free energy as well as heat and mass fluxes.
Through numerical benchmark simulations, we showcase how the parameters entering the model control temperature and density evolution during dendritic solidification processes and further verify the positive entropy-production property numerically.
Finally, we extend our model formulation to include the description of external mass and heat fluxes that enter the system,  thus mimicking the physics of open systems.
This work not only extends the existing non-isothermal PFC formulations~\cite{wang2021thermodynamically,punke2022explicit} but also allows for thermodynamic consistent PFC simulations of realistic solidification conditions incorporating thermal expansion of the crystal lattice and including heat fluxes. Therefore, our proposed model sets the ground for a comprehensive approach to crystal growth at diffusive time scales. The proposed simulations also offer accessible benchmarks for PFC and other solidification modeling and simulation approaches. 
The extension to three-dimensional systems can be envisaged with the aid of efficient numerical methods~\cite{TANG2011146,Yamanaka2017} or coarse-grained formulations of the PFC, like the so-called amplitude expansion of the PFC model~\cite{goldenfeld2005renormalization,athreya2006renormalization,salvalaglio2022coarse}.
Additionally, explicitly modeling elastic relaxation can provide a better description of competitive time scales, particularly regarding elastic relaxation and diffusive dynamics~\cite{skogvoll2022hydrodynamic,HeinonenPRL2016,SkaugenPRL2018}, along with the presence of heat flux.

\addcontentsline{toc}{section}{Data availability statement}
\section*{Data availability statement}
The data that support the findings of this study will be made openly available in suitable repositories in the final version.

\addcontentsline{toc}{section}{Acknowledgements}
\section*{Acknowledgements}
MP and MS acknowledge support from the Emmy Noether Programme of the Deutsche Forschungsgemeinschaft (DFG, German Research
Foundation) – Project no. 447241406. AV and MS acknowledge support from the Deutsche Forschungsgemeinschaft (DFG, German Research
Foundation) within FOR3013 – Project no. 417223351. SMW gratefully acknowledges support from the US National Science Foundation under grant NSF-DMS 2309547. Computing resources have been provided by the Center for Information Services and High-Performance Computing (ZIH), the \href{www.nhr-verein.de/unsere-partner}{NHR Center} of TU Dresden.

\addcontentsline{toc}{section}{ORCID iDs}
\section*{ORCID iDs}
Maik Punke \orcidlink{0000-0002-3564-7942} \href{https://orcid.org/0000-0002-3564-7942}{https://orcid.org/0000-0002-3564-7942}\\
Marco Salvalaglio \orcidlink{0000-0002-4217-0951} \href{https://orcid.org/0000-0002-4217-0951}{https://orcid.org/0000-0002-4217-0951} \\
Axel Voigt \orcidlink{0000-0003-2564-3697} \href{https://orcid.org/0000-0003-2564-3697}{https://orcid.org/0000-0003-2564-3697}\\
Steven M. Wise \orcidlink{0000-0003-3824-2075} \href{https://orcid.org/0000-0003-3824-2075}{https://orcid.org/0000-0003-3824-2075}

\addcontentsline{toc}{section}{References}
\section*{References}
\bibliographystyle{iopart-num-mod} 
\bibliography{references.bib}

\end{document}